\documentclass[11pt]{article}
\usepackage{amsmath,amssymb,amsthm}
\usepackage{booktabs}
\usepackage[margin=2.5cm]{geometry}
\usepackage[hidelinks]{hyperref}

\newtheorem{theorem}{Theorem}[section]
\newtheorem{lemma}[theorem]{Lemma}
\newtheorem{proposition}[theorem]{Proposition}
\newtheorem{corollary}[theorem]{Corollary}
\theoremstyle{remark}
\newtheorem{remark}[theorem]{Remark}
\newtheorem{example}[theorem]{Example}

\newcommand{\F}{\mathbb{F}}
\newcommand{\Z}{\mathbb{Z}}
\newcommand{\wt}{\operatorname{wt}}
\newcommand{\one}{\mathbf{1}}
\newcommand{\tr}{\operatorname{tr}}
\newcommand{\sgn}{\operatorname{sgn}}
\newcommand{\supp}{\operatorname{supp}}
\newcommand{\cC}{\mathcal{C}}
\newcommand{\bc}{\mathbf{c}}
\newcommand{\bv}{\mathbf{v}}
\newcommand{\bx}{\mathbf{x}}
\newcommand{\by}{\mathbf{y}}
\newcommand{\br}{\mathbf{r}}

\title{An integer determinant approach to the length of Type II codes and the weights of shadows}
\author{Jae-Hyun Baek\\[2pt]
\small Department of Artificial Intelligence and\\
\small Institute for Mathematical and Data Sciences\\
\small Sogang University, Seoul, Korea \\
\small \texttt{jhbaek@sogang.ac.kr}
\and
Jon-Lark Kim{\footnote{Corresponding author}}\\[2pt]
\small Department of Mathematics and\\
\small Institute for Mathematical and Data Sciences\\
\small Sogang University, Seoul, Korea\\
\small \texttt{jlkim@sogang.ac.kr}}
\date{}

\begin{document}
\maketitle

\begin{abstract}
It is well known that the length $n$ of a binary Type II code is divisible by $8$.
The classical proof relies on Gleason's theorem on weight enumerators.
Motivated by a question of Vera Pless, we give elementary proofs of this fact that avoid invariant theory.
Our main tool is a congruence modulo $8$ for the determinant of the integer Gram matrix $AA^{\mathsf T}$ of a binary orthogonal matrix $A$, obtained directly from the Leibniz expansion.
It gives a short proof of a theorem of Janusz: the number of rows of weight congruent to $3 \pmod 4$ in a binary orthogonal matrix is divisible by $4$.
From this single statement we derive both the length restriction for Type II codes and the theorem of Conway and Sloane that every vector in the shadow of a binary self-dual code of length $n$ has weight congruent to $n/2 \pmod 4$.
We also give two short character-theoretic proofs of the length restriction and explain how all of these arguments detect the Arf invariant of a quadratic form.
\end{abstract}

%%%%%%%%%%%%%%%%%%%%%%%%%%%%%%%%%%%%%%%%%%%%%%%%%%%%%%
%                   6. BODY
%%%%%%%%%%%%%%%%%%%%%%%%%%%%%%%%%%%%%%%%%%%%%%%%%%%%%%

%====================================================
\section{Introduction}
%====================================================

A binary linear code $\cC\subseteq\F_2^n$ is \emph{self-dual} if $\cC=\cC^\perp$.
It is \emph{doubly-even} if every codeword has weight divisible by $4$.
A binary doubly-even self-dual code is called a \emph{Type II} code.

\begin{theorem}[\cite{Gleason1971}]\label{thm:main}
If $\cC\subseteq\F_2^n$ is a binary Type II code of length $n$, then $8\mid n$.
\end{theorem}

Theorem~\ref{thm:main} is the basic arithmetic constraint in the theory of Type II codes.
The condition $8\mid n$ is also sufficient: the $[8,4,4]$ extended Hamming code $e_8$ is a Type II code, and so is the direct sum of $n/8$ copies of $e_8$ for every $n\equiv0\pmod8$.
Thus Type II codes exist precisely at the lengths divisible by $8$, and the main results on Type II codes are stated for these lengths.

Examples are the mass formula for the number of Type II codes of length $n$~\cite[Section~2]{RainsSloane1998} and the classification for small lengths: there are $1$, $2$, $9$ and $85$ inequivalent Type II codes of lengths $8$, $16$, $24$ and $32$, respectively~\cite[Table~I]{RainsSloane1998}, and $94343$ of length $40$~\cite{Bet-et-al}.

Another is the Mallows--Sloane bound $d\le4\lfloor n/24\rfloor+4$ on the minimum distance~\cite[Theorem~26]{RainsSloane1998}.
Extremal Type II codes of length $24m$, such as the $[24,12,8]$ extended Golay code, have the further property that their codewords of each fixed weight form $5$-designs~\cite[Theorem~36]{RainsSloane1998}.
Finally, Type II codes are closely tied to even unimodular lattices through Construction~A; for example, $e_8$ gives the $E_8$ lattice (see Section~\ref{sec:remarks}).

The standard proof of Theorem~\ref{thm:main} goes through Gleason's theorem~\cite{Gleason1971}: the weight enumerator of a Type II code is a polynomial in the weight enumerators of the $[8,4,4]$ extended Hamming code and the $[24,12,8]$ extended Golay code, which are homogeneous of degrees $8$ and $24$.
This argument determines the whole ring of possible weight enumerators, which is much more than the length restriction requires.
In the early 2000s, Vera Pless used to ask the second author whether there was an elementary way to prove Theorem~\ref{thm:main} (personal communication).
This question motivates the present paper.

Our main observation is that Theorem~\ref{thm:main} follows from an elementary congruence for integer Gram matrices (Sections~\ref{sec:det}).
Let $A$ be a binary orthogonal matrix, regarded as a $0/1$ integer matrix.
Modulo $8$, the Leibniz expansion of $\det(AA^{\mathsf T})$ reduces to the identity and the transpositions, and the total entry sum of $AA^{\mathsf T}$ is controlled by the column weights of $A$.
Comparing the result with $(\det A)^2\equiv1\pmod8$ gives a short proof of a theorem of Janusz~\cite{Janusz2007}: the number of rows of $A$ of weight $\equiv3\pmod4$ is divisible by $4$ (Proposition~\ref{prop:janusz}).
Applied to a generator matrix $[I_k\mid A]$ of a self-dual code, this single statement yields Theorem~\ref{thm:main}.
It also yields the theorem of Conway and Sloane~\cite{ConwaySloane1990} that every vector in the shadow of a self-dual code of length $n$ has weight $\equiv n/2\pmod4$ (Theorem~\ref{thm:shadow}); here the shadow has an explicit representative read off from the rows of $A$.

For comparison, we also give two character-theoretic proofs of Theorem~\ref{thm:main}, one via a Gauss sum (Section~\ref{sec:gauss}) and one via the MacWilliams identity~\cite{MacWilliams1963} evaluated at a single point (Section~\ref{sec:mw}).
All of these arguments detect the same obstruction, namely the Arf invariant of the quadratic form $\wt(\bx)/2 \bmod 2$ on the even-weight vectors (Remark~\ref{rem:arf}).
Some other arguments that avoid Gleason's polynomials can be found in the literature (see Section~\ref{sec:related}), but they rely on other machinery, such as spinor groups, the Witt--Arf theory of quadratic forms, or group-order computations in the orthogonal group.
Our contribution is an elementary and unified route to these results.
Section~\ref{sec:remarks} contains concluding remarks.

\subsection{Related work}\label{sec:related}

Several arguments for Theorem~\ref{thm:main} that avoid Gleason's polynomials are already in the literature.
We collect them here and indicate how they relate to our proofs.

%\paragraph{Invariant theory without the full ring of invariants.}
In the survey of Rains and Sloane~\cite[Section~6]{RainsSloane1998}, the Molien series of Gleason's group of order $192$ is computed as $1/((1-\lambda^8)(1-\lambda^{24}))$.
The authors remark that this alone already shows that the length of a doubly-even self-dual code is a multiple of $8$.
Elkies~\cite{Elkies2000} gives a shorter version of this idea.
Gleason's group contains the scalar matrices $\zeta I$ for all $8$th roots of unity $\zeta$.
Hence every homogeneous invariant has degree divisible by $8$.
Our third proof is a one-point form of this argument; see Section~\ref{sec:remarks}.

%\paragraph{Orthogonal geometry and spinor groups.}
Wood~\cite[Corollary~3.24]{Wood1989} proves Theorem~\ref{thm:main} by identifying doubly-even self-orthogonal codes with elementary abelian $2$-subgroups of $\mathrm{Spin}(n)$.
In coding terms, he uses the quadratic form $\wt(\bx)/2 \bmod 2$ on the even-weight vectors, which is our form $q$ in Section~\ref{sec:gauss}.
By the theorems of Witt and Arf, all maximal totally singular subspaces have the same dimension.
Wood computes this dimension from explicit examples.
It equals $n/2$, for even $n$, only when $n\equiv0\pmod8$.
Rains and Sloane~\cite[Section~2]{RainsSloane1998} likewise view a Type II code as a maximal totally singular subspace of an orthogonal geometry of dimension $n-2$.
Our second proof replaces the Witt--Arf theory by a single Gauss sum.

%\paragraph{Orthogonal matrices.}
Janusz~\cite{Janusz2007} studies the group $\mathrm O_m$ of $m\times m$ orthogonal matrices over $\F_2$.
He also studies its subgroup $\mathrm O_m^{(1)}$ of matrices all of whose rows have weight $\equiv1\pmod4$.
By computing the index $[\mathrm O_m:\mathrm O_m^{(1)}]$ and explicit coset representatives, he proves the following~\cite[Corollary~5]{Janusz2007}.
For every $A\in\mathrm O_m$, the number of rows of $A$ of weight $\equiv3\pmod4$ is a multiple of $4$.
Applied to a generator matrix $[I_k\mid A]$, this gives Theorem~\ref{thm:main}~\cite[Corollaries~6 and~7]{Janusz2007}.
Our first proof starts from the same reduction to an orthogonal matrix $A$.
It then proves the needed case by a determinant computation instead of a group-order computation.
The same computation also gives Janusz's general statement (Proposition~\ref{prop:janusz}); this is a different proof of a known result, not a new theorem.

%\paragraph{Determinants modulo $8$.}
In their proof that Type II codes lift to self-dual codes over $\Z_4$ with all norms divisible by $8$~\cite[Theorem~3(b)]{RainsSloane1998}, Rains and Sloane use a $\Z_4$-lift $\hat A$ of the matrix $A$ in a generator matrix $[I\mid A]$.
They use the congruence $\det(\hat A\hat A^{\mathsf T})\equiv1+4\tr\Delta\pmod8$, where $\Delta=\tfrac14(\hat A\hat A^{\mathsf T}+I)$.
That argument concerns the lifting problem, not the length of the code.
However, it uses the same kind of determinant computation modulo $8$ as our first proof.

\subsection{Notation and preliminaries}\label{sec:prelim}

Throughout, $\bx\cdot \by=\sum_j x_jy_j\in\F_2$ denotes the standard inner product.
We write $\bx\cap \by$ for the intersection of the supports of $\bx$ and $\by$, and $\one$ for the all-ones vector.
For $\bx,\by\in\F_2^n$,
\begin{equation}\label{eq:wtadd}
\wt(\bx+\by)=\wt(\bx)+\wt(\by)-2|\bx\cap \by|,\qquad |\bx\cap \by|\equiv \bx\cdot \by\pmod 2.
\end{equation}

Let $\cC$ be a Type II code of length $n$.
Since $\cC=\cC^\perp$, we have $\dim \cC=n/2$, so $n$ is even and $|\cC|=2^{n/2}$.
Every codeword has even weight, so $\one\cdot \bc=0$ for all $\bc\in \cC$.
Hence $\one\in \cC^\perp=\cC$.

\begin{lemma}\label{lem:char}
For a subspace $\cC$ of $\F_2^n$ and $\bv\in\F_2^n$,
\[
\sum_{\bc\in \cC}(-1)^{\bv\cdot \bc}=
\begin{cases}
|\cC|, & \bv\in \cC^\perp,\\
0, & \bv\notin \cC^\perp.
\end{cases}
\]
\end{lemma}

\begin{proof}
If $\bv\in \cC^\perp$, every term equals $1$.
If $\bv\notin \cC^\perp$, the map $\bc\mapsto \bv\cdot \bc$ is a surjective linear map $\cC\to\F_2$.
Hence it takes each of the values $0$ and $1$ exactly $|\cC|/2$ times, and the sum is $0$.
\end{proof}

%====================================================
\section{First proof: an integer Gram determinant}\label{sec:det}
%====================================================

This proof uses only a generator matrix, an integer Gram matrix $B$, its determinant modulo $8$, and the total sum of its entries.
As in~\cite{Janusz2007}, we first reduce to a statement about an orthogonal matrix over $\F_2$.

\begin{proof}[Proof of Theorem~\ref{thm:main}]
Put $k=n/2=\dim \cC$. The following steps prove the theorem.

\medskip\noindent
\textbf{Step 0: $k$ is even.}
We have $\one\in \cC$, since every codeword has even weight and $\cC$ is self-dual (Section~\ref{sec:prelim}). Hence $n=\wt(\one)\equiv0\pmod4$.

\medskip\noindent
\textbf{Step 1: The integer Gram matrix.}
After a coordinate permutation, $\cC$ has a generator matrix $[I_k\mid A]$.
Self-duality gives $AA^{\mathsf T}=I_k$ over $\F_2$.
Row $i$ of the generator has weight $1+\wt(\br_i)$, where $\br_i$ is row $i$ of $A$.
So the doubly-even condition implies $\wt(\br_i)\equiv3\pmod4$.

Regard $A$ as a $0/1$ integer matrix and put $B=AA^{\mathsf T}\in M_k(\Z)$.
Then
\[
B_{ii}=\wt(\br_i)\equiv3\pmod4,\qquad
B_{ij}=|\supp \br_i\cap\supp \br_j|\equiv0\pmod2\quad(i\ne j).
\]
Since $A$ is invertible modulo $2$, $\det A$ is odd.
Hence
\begin{equation}\label{eq:sq}
\det B=(\det A)^2\equiv1\pmod8.
\end{equation}

\medskip\noindent
\textbf{Step 2: Leibniz expansion for the determinant modulo $8$.}
Write
$$\det B=\sum_{\sigma\in S_k}\sgn(\sigma)\prod_iB_{i\sigma(i)}.$$
In the term of $\sigma$, a fixed index ($\sigma(i)=i$) contributes an odd diagonal factor $B_{ii}$.
A moved index ($\sigma(i)\ne i$) contributes an even off-diagonal factor $B_{i\sigma(i)}$.
So the number of even factors equals the number of indices moved by $\sigma$.
\begin{itemize}
\item If $\sigma$ moves at least three indices, its term is divisible by $2^3=8$ and vanishes modulo $8$.
\item No permutation moves exactly one index.
\item What remains is the identity, contributing $\prod_iB_{ii}$, and the transpositions $(ij)$.
\end{itemize}
Since $(ij)=(ji)$, each transposition corresponds to one unordered pair, which we index by $i<j$.
Its sign is $-1$, and by symmetry its term is
\[
-\,B_{ij}B_{ji}\prod_{\ell\ne i,j}B_{\ell\ell}=-\,B_{ij}^2\,D_{ij},
\qquad D_{ij}:=\prod_{\ell\ne i,j}B_{\ell\ell}\ \text{(odd)}.
\]
Write $B_{ij}=2t$.
Then it is easy to see that $B_{ij}^2=4t^2\equiv4t=2B_{ij}\pmod8$, since $t^2\equiv t\pmod2$.
Moreover, a multiple of $4$ is unchanged modulo $8$ by an odd factor, because $4t(2s+1)=4t+8ts$.
Hence, since $D_{ij}$ is odd we have $-B_{ij}^2D_{ij}\equiv-2B_{ij}\pmod8$, and
\begin{equation}\label{eq:leib}
\det B\equiv\prod_iB_{ii}-2\sum_{i<j}B_{ij}\pmod8.
\end{equation}

\medskip\noindent
\textbf{Step 3: The all-ones vector.}
We compute $\one^{\mathsf T}B\one$ through the columns of $A$, working in $\Z$.
Let $s_j=\sum_iA_{ij}$ be the weight of column $j$ of $A$, so that $A^{\mathsf T}\one=(s_1,\dots,s_k)^{\mathsf T}$.
Then
\begin{equation}\label{eq:ones}
\one^{\mathsf T}B\one=\one^{\mathsf T}A A^{\mathsf T} \one = (A^{\mathsf T}\one)^{\mathsf T}(A^{\mathsf T}\one)=\|A^{\mathsf T}\one\|^2=\sum_{j=1}^k s_j^2\equiv k\pmod8.
\end{equation}
The last congruence holds because every $s_j$ is odd.
Indeed, over $\F_2$ the square matrix $A$ satisfies $A^{-1}=A^{\mathsf T}$, so $A^{\mathsf T}A=I_k$.
The $(j,j)$ entry of $A^{\mathsf T}A$ is $s_j\bmod2$, so $s_j$ is odd.
For an odd $s=2u+1$ we have $s^2=4u(u+1)+1\equiv1\pmod8$, since $u(u+1)$ is even.

On the other hand, summing the entries of $B$ over all ordered pairs and using $B_{ij}=B_{ji}$ gives $\one^{\mathsf T}B\one= \sum_{i,j} B_{ij} = \tr B+2\sum_{i<j}B_{ij}$.
Substituting into Equation \eqref{eq:leib} gives
\begin{equation}\label{eq:key}
\det B\equiv\prod_iB_{ii}+\sum_iB_{ii}-k\pmod8.
\end{equation}

\medskip\noindent
\textbf{Step 4: Evaluate with all diagonal entries equal to $3$.}
Let $P(d_1,\dots,d_k)=\prod_id_i+\sum_id_i$ with all $d_i\equiv3\pmod4$.
Replacing one $d_i$ by $d_i+4$ changes the product by $4\cdot(\text{odd})\equiv4$ and the sum by $4$.
So $P\bmod8$ is unchanged, and we may take every $d_i=3$:
\[
\prod_iB_{ii}+\sum_iB_{ii}\equiv3^k+3k\pmod8.
\]
Combining this with Equations \eqref{eq:sq} and \eqref{eq:key}, and using $3^k\equiv1\pmod 8$ for even $k$ (Step~0), we get
\[
1\equiv3^k+3k-k\equiv1+2k\pmod8.
\]
Hence $8\mid2k$, so $4\mid k$ and $8\mid n$.
\end{proof}

\begin{remark}
Step~0 cannot be dropped: $k\equiv3\pmod4$ also satisfies $3^k+2k\equiv1\pmod8$, since $27+6=33\equiv1$.
\end{remark}

Steps~1--4 show that if every row of $A\in\mathrm O_k$ has weight $\equiv3\pmod4$, then $4\mid k$.
The same determinant computation proves the following general statement, due to Janusz~\cite[Corollary~5]{Janusz2007}.
We do not claim it as a new result.
We include it to give a direct determinant proof in place of the group-order computation of~\cite{Janusz2007}.

\begin{proposition}[Janusz~{\cite[Corollary~5]{Janusz2007}}]\label{prop:janusz}
Let $A\in\mathrm O_k$, i.e.\ $A\in\F_2^{k\times k}$ with $AA^{\mathsf T}=I_k$.
Then the number of rows of $A$ of weight $\equiv3\pmod4$ is a multiple of $4$.
\end{proposition}

\begin{proof}
Regard $A$ as a $0/1$ integer matrix and put $B=AA^{\mathsf T}\in M_k(\Z)$, as in Step~1.
Steps~2 and~3 use only three facts: the diagonal entries of $B$ are odd, its off-diagonal entries are even, and the columns of $A$ have odd weight.
All three hold for every $A\in\mathrm O_k$.
Hence \eqref{eq:sq}, \eqref{eq:ones} and \eqref{eq:key} remain valid:
\[
\det B\equiv1,\qquad \one^{\mathsf T}B\one\equiv k,\qquad
\det B\equiv\prod_iB_{ii}+\sum_iB_{ii}-k\pmod8.
\]
Write $d_i=B_{ii}=\wt(\br_i)=1+2\varepsilon_i+4h_i$ with $\varepsilon_i\in\{0,1\}$ and $h_i\in\Z$.
Then $a:=\sum_i\varepsilon_i$ is the number of rows of weight $\equiv3\pmod4$.

\emph{First we show that $a$ is even.}
Since $\tr B=\sum_id_i=k+2a+4\sum_ih_i$, every off-diagonal entry is even, and $B$ is symmetric, we have
\[
k+2a\equiv\tr B\equiv\tr B+2\sum_{i<j}B_{ij}=\one^{\mathsf T}B\one\equiv k\pmod4.
\]
Hence $2a\equiv0\pmod4$, so $a$ is even.

\emph{Second, we show that $4\mid a$.}
As in Step~4, replacing an odd $d_i$ by $d_i+4$ changes $\prod_id_i$ by $4\cdot(\text{odd})\equiv4$ and $\sum_id_i$ by $4$.
So $\prod_id_i+\sum_id_i\bmod8$ is unchanged, and we may take $d_i=1+2\varepsilon_i\in\{1,3\}$:
\[
\prod_iB_{ii}+\sum_iB_{ii}\equiv3^a+(k-a)+3a=3^a+k+2a\pmod8.
\]
Substituting into \eqref{eq:key} and using $3^a\equiv1\pmod8$ since $a$ is even,
\[
1\equiv\det B\equiv3^a+2a\equiv1+2a\pmod8.
\]
Hence $8\mid2a$, that is, $4\mid a$. This completes the proof.
\end{proof}

\begin{example}\label{ex:hamming}
Let $k=4$ and $A=J_4-I_4$, where $J_4$ is the $4\times4$ all-ones matrix, so that $[I_4\mid A]$ generates the $[8,4,4]$ extended Hamming code.
Every row and column of $A$ has weight $3$, and two distinct rows share two coordinates.
Hence $B=AA^{\mathsf T}$ has diagonal entries $3$ and off-diagonal entries $2$:
\[
B=I_4+2J_4,\qquad\det B=9.
\]
The $24$ permutations of $S_4$, grouped by type, contribute as follows.
\begin{center}
\footnotesize
\renewcommand{\arraystretch}{1.25}
\setlength{\tabcolsep}{3pt}
\begin{tabular}{@{}lcccc@{}}
\hline
Type & Number & Moved & Term & Mod $8$\\
\hline
identity $e$ & $1$ & $0$ & $3\cdot3\cdot3\cdot3=81$ & $1$\\
transp., e.g.\ $(12)$ & $6$ & $2$ & $-B_{12}B_{21}B_{33}B_{44}=-2\cdot2\cdot3\cdot3=-36$ & $4$\\
$3$-cycle, e.g.\ $(123)$ & $8$ & $3$ & $+B_{12}B_{23}B_{31}B_{44}=2\cdot2\cdot2\cdot3=24$ & $0$\\
double transp., e.g.\ $(12)(34)$ & $3$ & $4$ & $+2^4=16$ & $0$\\
$4$-cycle, e.g.\ $(1234)$ & $6$ & $4$ & $-2^4=-16$ & $0$\\
\hline
\end{tabular}
\end{center}

\begin{itemize}
\item \emph{Full sum:}
$\det B=81-6\cdot36+8\cdot24+3\cdot16-6\cdot16=81-216+192+48-96=9$.
\item \emph{Terms moving at least three indices:}
$192+48-96=144=8\cdot18\equiv0\pmod8$.
\item \emph{Identity and transpositions only:}
$81-216=-135\equiv1\pmod8$, which agrees with $\det B=9\equiv1$.
\item \emph{Removing the odd factor:}
for $(12)$, $D_{12}=B_{33}B_{44}=9$, and $-36\equiv4\equiv-4=-2B_{12}\pmod8$.
\item \emph{Total entry sum:}
column weights give $\sum_js_j^2=4\cdot9=36\equiv4=k\pmod8$.
Adding the entries of $B$ directly gives $4\cdot3+12\cdot2=36$ as well.
\item \emph{Formula \eqref{eq:key}:}
$\prod B_{ii}+\tr B-k=81+12-4=89\equiv1\pmod8$.
\end{itemize}
\end{example}

%====================================================
\section{Second proof: a quadratic form and a Gauss sum}\label{sec:gauss}
%====================================================

Let
\[
E=\{\bx\in\F_2^n : \wt(\bx)\equiv 0 \pmod 2\}
\]
be the space of even-weight vectors.
Every codeword of $\cC$ has even weight, so $\cC\subseteq E$.

\paragraph{The quadratic form.}
For $\bx\in E$ define $q(\bx)=\tfrac{\wt(\bx)}{2}\bmod 2\in\F_2$.
By Equation \eqref{eq:wtadd}, for $\bx,\by\in E$,
\begin{equation}\label{eq:quad}
q(\bx+\by)=q(\bx)+q(\by)+\bx\cdot \by .
\end{equation}
Thus $q$ is a quadratic form on $E$ whose polar form is the standard inner product.
The doubly-even condition says precisely that $q|_{\cC}\equiv 0$.
Moreover, by definition, $(-1)^{q(\bx)}=i^{\wt(\bx)}$.

\paragraph{The Gauss sum.}
Consider
\[
S=\sum_{\bx\in E}(-1)^{q(\bx)}=\sum_{\bx\in E} i^{\wt(\bx)}.
\]

\emph{(a) Computation via the code.}
Write $E$ as the disjoint union of the cosets $\bv+\cC$.
Since $q(\bc)=0$ for $\bc\in \cC$, identity \eqref{eq:quad} gives
\[
\sum_{\bc\in \cC}(-1)^{q(\bv+\bc)}
=(-1)^{q(\bv)}\sum_{\bc\in \cC}(-1)^{\bv\cdot \bc}.
\]
By Lemma~\ref{lem:char}, this equals $(-1)^{q(\bv)}|\cC|$ if $\bv\in \cC^\perp=\cC$, and $0$ otherwise.
Hence the only coset that contributes is $\cC$ itself.
Taking $\bv\in \cC$, we have $q(\bv)=0$, so
\[
S=|\cC|=2^{n/2}.
\]

\emph{(b) Direct computation.}
There are $\binom{n}{w}$ vectors of weight $w$, so by the binomial theorem
\[
S=\sum_{w\ \text{even}}\binom{n}{w}i^{w}
=\operatorname{Re}\sum_{w=0}^{n}\binom{n}{w}i^{w}
=\operatorname{Re}\,(1+i)^n
=2^{n/2}\cos\frac{n\pi}{4}.
\]
Here we used that $i^w$ is purely imaginary for odd $w$, and that $1+i=\sqrt2\,e^{i\pi/4}$.

\paragraph{Conclusion.}
Comparing (a) and (b) gives $2^{n/2}\cos\frac{n\pi}{4}=2^{n/2}$.
Hence $\cos\frac{n\pi}{4}=1$, which is equivalent to $8\mid n$.
This proves Theorem~\ref{thm:main}. \qed

\medskip

\begin{remark}\label{rem:arf}
Since $n$ is even, $\one\in E$ and $\one\in E^\perp$.
Also $\one\in \cC$, so $q(\one)=0$, and therefore $q$ descends to $\overline E=E/\langle\one\rangle$.
This is a nondegenerate quadratic form on a space of dimension $n-2$.
For a nondegenerate quadratic form $Q$ on a $2m$-dimensional space $V$ over $\F_2$, the Arf invariant $\operatorname{Arf}(Q)\in\F_2$~\cite{Arf1941},~\cite{CimasoniSimian} is characterized by
\[
\sum_{\bx\in V}(-1)^{Q(\bx)}=(-1)^{\operatorname{Arf}(Q)}\,2^m.
\]
The subspace $\overline \cC=\cC/\langle\one\rangle$ is totally singular of dimension $n/2-1$, i.e.,\ a \emph{totally singular Lagrangian}.
In characteristic $2$ this is stronger than being Lagrangian for the polar form, since the polar form is alternating.

If a totally singular Lagrangian exists, then the Arf invariant of the form is $0$:
this is exactly the coset argument in (a), which shows that the Gauss sum is positive.
(The converse also holds, but we do not need it.)
On the other hand, for $n\equiv0\pmod4$ the Arf invariant of $\overline E$ is $n/4 \bmod 2$.
Indeed, for $\bx\in E$ we have $q(\bx+\one)=q(\bx)+q(\one)+\bx\cdot\one=q(\bx)$.
So the sum $S$ over $E$ is twice the corresponding sum over $\overline E$, and by (b) the latter equals $2^{(n-2)/2}\cos\frac{n\pi}{4}=(-1)^{n/4}\,2^{(n-2)/2}$.
%For example, when $n=4$, the three nonzero elements $\mathbf{1100},\mathbf{1010},\mathbf{1001}$ of $\overline E$ all satisfy $q=1$.
%Thus the line $\langle\mathbf{1100}\rangle$ is Lagrangian for the polar form but is not totally singular, and $\overline E$ has no totally singular Lagrangian.
%The Gauss sum computation above carries out this Arf invariant test directly.
%This orthogonal-geometry point of view is the one behind~\cite[Corollary~3.24]{Wood1989} and~\cite[Section~2]{RainsSloane1998}; see Section~\ref{sec:related}.
\end{remark}

%====================================================
\section{Third proof: the MacWilliams identity at one point}\label{sec:mw}
%====================================================

Let
\[
W_{\cC}(x,y)=\sum_{\bc\in \cC}x^{\,n-\wt(\bc)}y^{\wt(\bc)}
\]
be the weight enumerator of $\cC$.
The MacWilliams identity~\cite{MacWilliams1963} states
\[
W_{\cC^\perp}(x,y)=\frac{1}{|\cC|}\,W_{\cC}(x+y,\;x-y).
\]
Since $\cC=\cC^\perp$ and $|\cC|=2^{n/2}$,
\begin{equation}\label{eq:mw}
W_{\cC}(x,y)=2^{-n/2}\,W_{\cC}(x+y,\;x-y).
\end{equation}

Now we substitute $(x,y)=(1,i)$ into Equation \eqref{eq:mw}.

\emph{Left-hand side.}
Every $\bc\in \cC$ satisfies $4\mid\wt(\bc)$, so $i^{\wt(\bc)}=1$, and hence
\[
W_{\cC}(1,i)=|\cC|=2^{n/2}.
\]

\emph{Right-hand side.}
$W_{\cC}$ is homogeneous of degree $n$, and $\dfrac{1-i}{1+i}=-i$.
Therefore
\[
W_{\cC}(1+i,\,1-i)=(1+i)^n\,W_{\cC}\!\left(1,\tfrac{1-i}{1+i}\right)=(1+i)^n\,W_{\cC}(1,-i).
\]
Since $(-i)^{\wt(\bc)}=1$, we have $W_{\cC}(1,-i)=2^{n/2}$, so the right-hand side equals
\[
2^{-n/2}\,(1+i)^n\,2^{n/2}=(1+i)^n .
\]

\paragraph{Conclusion.}
We obtain $(1+i)^n=2^{n/2}$.
Since $(1+i)^n=2^{n/2}e^{in\pi/4}$, it follows that $e^{in\pi/4}=1$, and therefore $8\mid n$.
This proves Theorem~\ref{thm:main}. \qed

%====================================================
\section{An application: the shadow of a self-dual code}\label{sec:shadow}
%====================================================

In this section $\cC$ is a binary self-dual code of length $n=2k$ that need not be doubly-even.
For $\bx,\by\in\cC$ we have $\bx\cdot\by=0$, so by \eqref{eq:quad} the map $q\colon\cC\to\F_2$, $q(\bc)=\wt(\bc)/2\bmod2$, is linear.
Its kernel $\cC_0$ is the doubly-even subcode of $\cC$.
Following Conway and Sloane~\cite{ConwaySloane1990}, the \emph{shadow} of $\cC$ is
\[
S(\cC)=
\begin{cases}
\cC_0^\perp\setminus\cC, & \text{if $\cC$ is not doubly-even},\\
\cC, & \text{if $\cC$ is doubly-even};
\end{cases}
\]
see also~\cite{RainsSloane1998}.
It is classical that every vector of $S(\cC)$ has weight $\equiv n/2\pmod4$~\cite{ConwaySloane1990}, \cite{RainsSloane1998}.
The usual proof uses the shadow weight enumerator $W_{\cC}(x+y,\,i(x-y))/|\cC|$.
We give a proof based on an explicit representative of the shadow and on Proposition~\ref{prop:janusz}.

\begin{theorem}[\cite{ConwaySloane1990}]\label{thm:shadow}
Let $\cC$ be a binary self-dual code of length $n$.
Then every vector of $S(\cC)$ has weight congruent to $n/2\pmod4$.
\end{theorem}

\begin{proof}
Put $k=n/2$.
After a coordinate permutation, $\cC$ has a generator matrix $[I_k\mid A]$ with rows $\mathbf g_i=(\mathbf e_i,\br_i)$, and $A\in\mathrm O_k$.
Each $\wt(\br_i)$ is odd, and $\wt(\mathbf g_i)=1+\wt(\br_i)$.
Define $\varepsilon\in\F_2^k$ by $\varepsilon_i=1$ if $\wt(\br_i)\equiv1\pmod4$ and $\varepsilon_i=0$ if $\wt(\br_i)\equiv3\pmod4$, and put $\mathbf s=(\varepsilon,\mathbf 0)\in\F_2^n$.

\emph{Step 1: $S(\cC)=\mathbf s+\cC$.}
We have $\mathbf s\cdot\mathbf g_i=\varepsilon_i=q(\mathbf g_i)$ for every $i$.
Both sides are linear on $\cC$, so $\mathbf s\cdot\bc=q(\bc)$ for all $\bc\in\cC$.
If $\cC$ is doubly-even, then $\varepsilon=\mathbf 0$ and $\mathbf s+\cC=\cC=S(\cC)$.
Otherwise $q\neq0$, so $\cC_0=\{\bc\in\cC:\mathbf s\cdot\bc=0\}$ has index $2$ in $\cC$ and $\mathbf s\notin\cC^\perp=\cC$.
Since $\cC_0\subseteq\cC=\cC^\perp$, we have $\cC\subseteq\cC_0^\perp$, and also $\mathbf s\in\cC_0^\perp$.
As $\dim\cC_0^\perp=n-(k-1)=k+1$, it follows that $\cC_0^\perp=\cC\cup(\mathbf s+\cC)$, hence $S(\cC)=\mathbf s+\cC$.

\emph{Step 2: $\wt(\mathbf s)\equiv k\pmod4$.}
Let $a$ be the number of rows of $A$ of weight $\equiv3\pmod4$.
Then $\wt(\mathbf s)=k-a$, and $4\mid a$ by Proposition~\ref{prop:janusz}.

\emph{Step 3: all shadow vectors.}
For $\bc\in\cC$ we have $|\mathbf s\cap\bc|\equiv\mathbf s\cdot\bc=q(\bc)\pmod2$, so $2|\mathbf s\cap\bc|\equiv\wt(\bc)\pmod4$.
By \eqref{eq:wtadd},
\[
\wt(\mathbf s+\bc)=\wt(\mathbf s)+\wt(\bc)-2|\mathbf s\cap\bc|\equiv\wt(\mathbf s)\equiv k\pmod4.
\]
This proves the theorem.
\end{proof}

\begin{remark}
If $\cC$ is a Type II code, then $\mathbf s=\mathbf 0$ and $S(\cC)=\cC$, and Theorem~\ref{thm:shadow} says $0\equiv n/2\pmod4$.
This is Theorem~\ref{thm:main} again.
In contrast to the proof of Theorem~\ref{thm:main} in Section~\ref{sec:det}, Theorem~\ref{thm:shadow} needs Proposition~\ref{prop:janusz} for an arbitrary number of rows of weight $\equiv3\pmod4$.
\end{remark}

The residue in Theorem~\ref{thm:shadow} gives the classical construction of Type II neighbours~\cite[Section~5]{RainsSloane1998}.

\begin{corollary}\label{cor:neighbours}
Let $\cC$ be a binary self-dual code of length $n\equiv0\pmod8$ that is not doubly-even, and let $\mathbf s$ be as in the proof of Theorem~\ref{thm:shadow}.
Then $\cC_0+\langle\mathbf s\rangle$ and $\cC_0+\langle\mathbf s+\bc\rangle$, for any $\bc\in\cC\setminus\cC_0$, are two distinct Type II codes of length $n$.
\end{corollary}

\begin{proof}
Let $\mathbf t\in\{\mathbf s,\mathbf s+\bc\}$.
Then $\mathbf t\in S(\cC)$, so $\wt(\mathbf t)\equiv n/2\equiv0\pmod4$ by Theorem~\ref{thm:shadow}, and in particular $\mathbf t\cdot\mathbf t=0$.
Also $\mathbf t\in\cC_0^\perp$.
Hence $\cC_0+\langle\mathbf t\rangle$ is self-orthogonal and is spanned by vectors of weight divisible by $4$, so it is doubly-even by \eqref{eq:wtadd}.
Its dimension is $(k-1)+1=k=n/2$, since $\mathbf t\notin\cC\supseteq\cC_0$; so it is self-dual.
Since $\mathbf t\notin\cC$, each of the two codes meets $\cC$ exactly in $\cC_0$.
If they were equal, this code would contain $\mathbf s$ and $\mathbf s+\bc$, hence $\bc\in\cC\setminus\cC_0$, a contradiction.
So the two codes are distinct.
\end{proof}

%====================================================
\section{Concluding remarks}\label{sec:remarks}
%====================================================

%\paragraph{The third proof.}
The first proof uses neither characters nor weight enumerators.
Its modulo-$8$ information comes from the fact that $\det B$ is an odd square.
The proof then evaluates $\det B$ modulo $8$ using only the diagonal entries and the total entry sum of $B$.
Compared with~\cite{Janusz2007}, it replaces the computation of $[\mathrm O_k:\mathrm O_k^{(1)}]$ by this determinant congruence.
Compared with~\cite[Theorem~3(b)]{RainsSloane1998}, the lift of $A$ is the plain $0/1$ lift, which is not corrected to be orthogonal modulo $4$.
%Remark~\ref{rem:arf} explains the common source of all three proofs.
The second proof interprets the length restriction through the Arf invariant, whereas the first detects it directly through an integer determinant congruence.
The obstruction is the Arf invariant of the quadratic form $q$ on $E/\langle\one\rangle$, which must vanish for a Type II code to exist.

%\paragraph{The first two proofs.}
Part (a) of the second proof is exactly the argument that proves the MacWilliams identity, namely Lemma~\ref{lem:char}, carried out at a single point.
So the second and third proofs are essentially the same computation, written with and without the weight enumerator.

%\paragraph{Relation to Gleason's group.}
Gleason's theorem rests on the invariance of $W_{\cC}$ under the group of order $192$ generated by
\[
T_1=\tfrac{1}{\sqrt2}\begin{pmatrix}1&1\\1&-1\end{pmatrix},\qquad
T_2=\begin{pmatrix}1&0\\0&i\end{pmatrix}.
\]
The third proof uses $T_1$ and $T_2$ once each, at one point, rather than the whole group.
A direct computation gives $(T_1T_2)^3=e^{i\pi/4}I$.
Since $W_{\cC}$ is homogeneous of degree $n$, invariance under this scalar matrix also forces $e^{in\pi/4}=1$; see~\cite{Elkies2000} and~\cite[Ch.~19]{MacWilliamsSloane1977}.

%\paragraph{Lattices.}
Construction~A sends a Type II code of length $n$ to a positive definite even unimodular lattice of rank $n$.
Theorem~\ref{thm:main} is therefore the coding analogue of the fact that the rank of an even unimodular lattice is divisible by $8$~\cite{Serre1973}.
The standard proofs of the lattice statement, via theta series or via signatures, are less elementary than the proofs above.

%It would be interesting to know whether the determinant argument of Section~\ref{sec:det} can be adapted to give an elementary proof that the rank of an even unimodular lattice is divisible by $8$.

%The acknowledgments section should not be numbered.
\section*{Acknowledgement}
This research (J.-L. Kim) was supported in part by the BK21 FOUR (Fostering Outstanding Universities for Research) funded by the Ministry of Education (MOE, Korea), National Research Foundation of Korea (NRF) under Grant No. 4120240415042, Basic Science Research Program through the National Research Foundation of Korea (NRF) funded by the Ministry of Science and ICT under Grant No. RS-2025-24534992 and Global - Learning \& Academic research institution for Master’s·PhD students, and Postdocs(LAMP) Program of the National Research Foundation of Korea(NRF) grant funded by the Ministry of Education(No. RS-2024-00441954).

%\section*{Use of Generative-AI tools declaration}
%The author(s) declare(s) they have used Artificial Intelligence (AI) tools in the creation of this article.
%AI tools used: Claude Opus 5.5.
%The AI tool was used to review the manuscript and check grammatical errors.

\section*{Use of Generative-AI tools declaration}
The author(s) declare(s) they have used Artificial Intelligence (AI) tools in the creation of this article.
AI tools used: Claude Opus 5.5(Anthropic). The tool was used to help draft parts of the text and proofs, and to perform numerical checks of the statements. All mathematical contents were verified by the authors, who take full responsibility for it.

\medskip
% The information below will be filled in by AIMS production staff.
%Received xxxx 20xx; revised xxxx 20xx; early access xxxx 20xx.
\medskip

\end{document}